\documentclass[aps,showpacs,pra,twocolumn,superscriptaddress]{revtex4-2}

\usepackage[utf8]{inputenc}
\usepackage[T1]{fontenc}
\usepackage[polish,english]{babel}
\usepackage{times}
\usepackage{amssymb}
\usepackage{mathtools}
\usepackage{amsthm}
\usepackage{MnSymbol}
\usepackage{dsfont}
\usepackage{bm}
\usepackage{textcomp}
\usepackage{xcolor}
\usepackage{physics}
\usepackage[shortlabels]{enumitem}
\usepackage{dcolumn}
\usepackage{graphicx}
\let\ChangesComment\comment
\usepackage{changes}
\let\comment\ChangesComment
\usepackage{comment}
\usepackage[colorlinks=true,citecolor=blue,urlcolor=blue]{hyperref}

\theoremstyle{plain}
\newtheorem{lemma}{Lemma}
\newtheorem{fact}{Fact}
\newtheorem{theorem}{Theorem}

\newtheorem{corollary}{Corollary}

\definecolor{darkgreen}{rgb}{0,0.5,0}

\begin{document}

\title{\texorpdfstring{On Non-Existence of Absolutely Maximally Entangled Canonical Graph States in Even Local Dimensions}{}}
\author{Jakub W\'{o}jcik}
\email{jwojcik@cft.edu.pl}
\affiliation{Center for Quantum-Enabled Computing, Center for Theoretical Physics, Polish Academy of Sciences, Aleja Lotnik\'{o}w 32/46, 02-668 Warsaw, Poland}
\author{Owidiusz Makuta}
\email{makuta@lorentz.leidenuniv.nl}
\affiliation{$\langle aQa ^L\rangle $ Advancing Quantum Algorithms, Universiteit Leiden}
\affiliation{Instituut-Lorentz, Universiteit Leiden, Niels Bohrweg 2, 2333 CA Leiden, Netherlands}
\author{Wojciech Bruzda}
\affiliation{Center for Quantum-Enabled Computing, Center for Theoretical Physics, Polish Academy of Sciences, Aleja Lotnik\'{o}w 32/46, 02-668 Warsaw, Poland}
\author{Remigiusz Augusiak}
\affiliation{Center for Quantum-Enabled Computing, Center for Theoretical Physics, Polish Academy of Sciences, Aleja Lotnik\'{o}w 32/46, 02-668 Warsaw, Poland}

\date{\today}

\begin{abstract}
We demonstrate that absolutely maximally entangled (AME) states consisting of $N=4n$ qudits with $n\in\{1,2,3,\ldots\}$, each of even local dimension $d$, cannot be realized as canonical graph states over the ring $\mathbb{Z}_d$. This result imposes strong constraints on AME states in composite local dimensions and characterizes the limitations of graph-state constructions for highly entangled multipartite quantum systems. Furthermore, for $d\equiv 2\pmod 4$ this obstruction, combined with the prime-power decomposition of stabilizer states, excludes all pure stabilizer AME states (including four quhexes), while clarifying the distinction between canonical $\mathbb{Z}_{d}$ graph states and stabilizer constructions over prime-power factors. We also discuss a composite construction of mixed $k$-uniform states. For $(N,d)=(4,6)$, the construction yields a rank-two mixed $2$-uniform stabilizer state of purity $1/2$, which is optimal among normalized stabilizer projectors.
\end{abstract}

\maketitle

\section{Introduction}

Absolutely maximally entangled (AME) states, defined as pure multipartite quantum states whose every balanced reduction is maximally mixed, represent one of the most extreme forms of multipartite entanglement~\cite{Goyeneche2015}. These states generalize the maximally entangled Bell states~\cite{HHHH09} to multi-party scenarios and play a central role in quantum information theory. AME states appear in many contexts including quantum error correction~\cite{KL97, Mazurek2020}, quantum secret sharing~\cite{Helwig2012} and even quantum gravity where, under the name \emph{perfect tensors}, they serve as fundamental building blocks of tensor-network models of the AdS/CFT correspondence~\cite{HaPPy_codes}.

A special subclass of AME states is formed by \emph{stabilizer states}, defined as simultaneous $+1$ eigenstates of an abelian subgroups of the generalized Pauli group~\cite{Gottesman1997,Hostens_2005}. Due to their algebraic structure and efficient classical description, the experimental realization of such states is generally more feasible compared with generic highly entangled ones. However, stabilizer circuits alone can be efficiently simulated on a classical computer~\cite{Gottesman1997}. Consequently, states outside this class, known as \emph{magic states}, are essential to achieve universal quantum computation and demonstrate quantum computational advantage~\cite{BK05}.

In this work we use the term \emph{canonical graph state} for graph states defined via the generalized Pauli operators (detailed in Section~\ref{sec:graph-state}) and an adjacency matrix with entries in the ring $\mathbb{Z}_d$. This allows us to translate graph-theoretic machinery into linear-algebra constraints over $\mathbb{Z}_d$. Nevertheless, this should be distinguished from stabilizer formalism over finite fields. For a prime local dimension ($d=p$), the ring $\mathbb{Z}_p$ is a finite field, while for prime powers and composite dimensions this is no longer true.
Hence, we formulate the main theorem as a non-existence result for canonical $\mathbb{Z}_d$ graph states.

Since the discovery of the first representative of an AME$(4,6)$ state~\cite{1stAME}, followed by a series of subsequent results~\cite{AME_BUV, Bruzda2023, Goedicke2023, Casas_2026}, it has been the subject of ongoing debate whether any such states belong to the class of stabilizer states. Although individual cases have been ruled out, a general consensus has been lacking. In the recent paper~\cite{Cha26}, the author presents a proof that a stabilizer AME$(4,6)$ state does not exist. The argument relies on structural properties of stabilizer states in dimension $6$, namely that every stabilizer state of local dimension six is locally unitarily equivalent to a tensor product of a qubit and a qutrit stabilizer state. Using this decomposition, it is shown that if a stabilizer AME$(4,6)$ state existed, it would necessarily contain an AME$(4,2)$ stabilizer state as a component. However, such states do not exist~\cite{HIGUCHI2000213, Huber_2018} which implies a contradiction. The work \cite{Cha26} further generalizes this argument to arbitrary composite local dimensions, concluding that no stabilizer AME$(4,d)$ state exists for $d\equiv 2\,(\bmod\,4)$.

In this paper, we establish an independent graph-state obstruction to the existence of a graph AME$(4,6)$ state, which originally motivated this study, and exclude an infinite family of AME states. Specifically, we prove that there is no canonical graph AME$(4n, 2\ell)$ state over $\mathbb{Z}_{2\ell}$ for any $n,\ell\in\mathbb{N}_{+}=\{1,2,3,\ldots\}$. Furthermore, for local dimensions $d\equiv 2 \pmod 4$, this obstruction extends to exclude all pure stabilizer AME$(4n,d)$ states. This distinction between ring-based graph states and general stabilizer states is crucial: in prime-power dimensions, stabilizer constructions over finite fields may exist even when canonical graph states over $\mathbb{Z}_d$ are forbidden. The approach presented here is complementary to Ref.~\cite{Cha26}. For example, Ref.~\cite{Cha26} excludes stabilizer AME$(7,6)$ states because pure qubit AME$(7,2)$ states do not exist (a scenario outside our $N=4n$ setting). Conversely, while stabilizer AME$(4,4)$ states exist over the field $\mathbb{F}_4$, our method directly proves that they cannot be realized as canonical graph states over the ring $\mathbb{Z}_4$.

The paper is organized as follows. In Section~\ref{sec:graph-state}, we introduce the formalism of graph states and give a formal definition of AME states. Section~\ref{sec:main-result} presents the main result of this work -- a theorem establishing the non-existence of canonical graph AME states with $4n$ parties, for $n\in\mathbb{N}_+$ and even local dimensions, as well as its implications for general stabilizer states. In Section~\ref{sec:mixed-state} we extend the analysis beyond pure states and provide a construction of mixed $k$-uniform states~\cite{PhysRevA.100.032112} with optimal purity. In the final section, we summarize our findings and discuss their implications and possible directions for future research.

\section{Preliminaries}\label{sec:graph-state}

A graph $G$ is a tuple of a set of vertices $V$, which we denote by $\{1,2,\ldots,N\}=[N]$, and a set of edges $E$ connecting them. Here, we focus on multigraphs, which are a generalisation of a standard graph definition that allows for $E$ to be a multiset, i.e., a single edge can appear more than once. However, for convenience, we will be referring to multigraphs as graphs.

With every graph $G$ one can associate an adjacency matrix $\Gamma$ that encodes the connectivity of $G$. The element $\Gamma_{ij}$ is equal to the number of times that an edge $(i,j)$ appears in $E$. As we consider undirected graphs with no loops, it implies $\Gamma_{ij}=\Gamma_{ji}$ and $\Gamma_{ii}=0$.

A graph state $\ket{G}\in (\mathbb{C}^{d})^{\otimes N}$ is a unique state satisfying $g_i|G\rangle =|G\rangle$ for all $i\in[N]$, where
\begin{equation}
    g_i=X_i\prod_{j=1}^N Z_j^{\Gamma_{ij}},
\end{equation}
with $X_i=\mathbb{I}^{\otimes(i-1)}\otimes X\otimes \mathbb{I}^{\otimes (N-i)}$ and
\begin{equation}\label{gen-d-paulis}
    X=\sum_{j=0}^{d-1}|j\oplus_d 1\rangle\!\langle j|,\qquad Z=\sum_{j=0}^{d-1}\omega^j|j\rangle\!\langle j|, \qquad \omega = \exp\frac{2\pi i}{d}.
\end{equation}
The stabilizer $\mathbb{S}$ of $|G\rangle$ is a group generated by operators $g_i$, i.e., $\mathbb{S}=\langle g_1, g_2, \ldots, g_N\rangle$. The projector onto $|G\rangle$ is given by 
\begin{equation}\label{eq: projector}
|G\rangle\!\langle G|=\frac{1}{d^N}\prod_{i=1}^N\left(\sum_{j=0}^{d-1}g_i^j\right)=\frac{1}{d^N}\sum_{S\in\mathbb{S}}S.
\end{equation}

An $N$-partite state $\ket{\psi}$ is absolutely maximally entangled (AME state) if
\begin{equation}
\forall_{\substack{Q\subset[N]\\ |Q|=\lceil N/2\rceil}}\quad   {\rm Tr}_{Q}\big(|\psi\rangle\langle \psi|\big)=\frac{1}{d^{\lfloor N/2\rfloor}}\mathbb{I}^{\otimes\lfloor N/2\rfloor},
\end{equation}
where $\mathbb{I}$ is the $d\times d$ identity matrix, ${\rm Tr}_{Q}(\cdot)$ is the partial trace over $Q$, and $\lfloor\cdot\rfloor$, $\lceil\cdot\rceil$ stand for the floor and ceiling functions, respectively. To make the notation more concise, we write ${\rm Tr}_{\lceil N/2 \rceil}$ to indicate that a condition holds for all subsets $Q \subset [N]$ with $|Q| = \lceil N/2 \rceil$. Then, a graph state $|G\rangle$ is an AME state if
\begin{equation}
  {\rm Tr}_{\lceil N/2 \rceil}\big(|G\rangle\langle G|\big)=\frac{1}{d^{\lfloor N/2\rfloor}}\mathbb{I}^{\otimes\lfloor N/2\rfloor},
\end{equation}
which, by virtue of Eq.~\eqref{eq: projector}, is equivalent to
\begin{equation}
  {\rm Tr}_{\lceil N/2 \rceil}\left(\sum_{S\in\mathbb{S}\setminus\{\mathbb{I}\}}S\right)=0.\label{AMEG}
\end{equation}

\section{Main Result}\label{sec:main-result}
Before proving the main result, we first introduce a simplified, necessary and sufficient condition for a graph state to be an AME state.
\begin{lemma}\label{lem: AME}
A graph state $\ket{G}$ stabilized by a stabilizer $\mathbb{S}$ is AME state iff
\begin{equation}
    \forall\,S\in\mathbb{S}\setminus\{\mathbb{I}\}\,\,\,:\,\,\,{\rm Tr}_{\lceil N/2\rceil}S=0.\label{AMES}
\end{equation}
\end{lemma}
\begin{proof}
Clearly, Eq. \eqref{AMES} implies Eq. \eqref{AMEG}, and so it is sufficient to prove the implication in the opposite direction. Let us assume condition~\eqref{AMEG} holds and that condition~\eqref{AMES} does not. Then there exists a subsystem $Q_0 \subset [N]$ with $|Q_0| = \lceil N/2 \rceil$ such that the set
\begin{equation}
A_{Q_{0}} := \left\{S\in\mathbb{S}\setminus\{\mathbb{I}\}: {\mathrm{Tr}}_{Q_{0}}(S)\neq 0\right\},
\end{equation}
is nonempty.
Since ${\rm Tr}_{Q_0}(S) = 0$ for all $S \notin A_{Q_0}$, Eq. \eqref{AMEG} reduces to
\begin{equation}\label{dodatkowe}
 \underset{S\in A_{Q_0}}{\sum}{\rm Tr}_{Q_0}(S) = 0.
\end{equation}
Since ${\rm Tr}_{Q_0}(S) \neq 0$, then any $S \in A_{Q_0}$ must act as the identity on $Q_0$, i.e., $S = S_{\overline{Q}_0} \otimes \mathbb{I}_{Q_0}$ where $\overline{Q}_0 = [N] \setminus Q_0$. Taking the trace over $Q_0$ yields ${\rm Tr}_{Q_0}(S) = d^{\lceil N/2 \rceil} S_{\overline{Q}_0}$, and thus
\begin{equation}
    \underset{S \in A_{Q_0}}{\forall}\,\, S=d^{-\lceil N/2\rceil}{\rm Tr}_{Q_0}(S) \otimes \mathbb{I}_{Q_0}.
\end{equation}
Summing over $S \in A_{Q_0}$ and using Eq.~\eqref{dodatkowe}, it follows that $\underset{S\in A_{Q_0}}{\sum}S=0$. This directly contradicts the linear independence of distinct Pauli operators in $\mathbb{S}$, or equivalently, $\underset{S \in A_{Q_0}}{\sum}S|G\rangle = |A_{Q_0}||G\rangle = 0$ implies $A_{Q_0}=\emptyset$. Contradiction.
\end{proof}

For even $N$, condition~\eqref{AMES} in Lemma~\ref{lem: AME} is precisely the minimum-distance requirement $d_{\text{min}}\geq N/2+1$ for a pure self-dual stabilizer code $\left[\!\left[N,0,N/2+1\right]\!\right]_{d}$. This is the standard connection between even-party AME states and pure quantum MDS-type stabilizer codes~\cite{Helwig2012,GR15_MDSC}.

\begin{comment}
For completeness we prove the following fact.
\begin{fact}\label{fact}
Let $A\in M_{N\times N}(\mathbb{Z}_{d})$ be a matrix over a finite ring $\mathbb{Z}_{d}$, where $d$ is an integer $d \geqslant 2$. If $\det(A)$ is not a unit in $\mathbb{Z}_{d}$ then $\ker(A)\neq \{0\}$.
\end{fact}
\begin{proof}
If $A$ is invertible then $\det(A)$ is a unit in $\mathbb{Z}_{d}$. Conversely, since $\det(A)$ is not a unit, $A$ is not invertible, therefore $A$ is not a bijection. Since $A$ is an endomorphism over a finite set of elements, and every such injective endomorphism is surjective, then if $A$ is not bijective, then it is not injective. Therefore, there exist two distinct vectors $x_{1},x_{2}\in \mathbb{Z}_{d}^{N}$ such that $A x_{1}=Ax_{2}$. This immediately implies that the non-zero vector $x_{1}-x_{2}$ is in the kernel of $A$, and so $\ker(A)\neq \{0\}$.
\end{proof}
{\color{red}To jest znane twierdzenie. Proponuję zostawić elementarny dowód i ewentualnie dopisać referencję do ogólniejszego wyniku~\cite{McCoy1942}. Nie pisałbym nic o AI, bo będziemy się musieli dodatkowo tłumaczyć przed recenzentami.}
\end{comment}
The following is the square-matrix case of McCoy's theorem~\cite{McCoy1942},
since every nonunit of $\mathbb{Z}_d$ has a nonzero annihilator.
\begin{fact}\label{fact}
Let $d\geqslant 2$ and $A\in M_{N\times N}(\mathbb{Z}_d)$.
If $\det(A)$ is not a unit in $\mathbb{Z}_d$, then
$\ker(A)\neq\{0\}$.
\end{fact}

With that, we can proceed to the main result of this work.
\begin{theorem}\label{thm}
There are no canonical $\mathbb{Z}_d$ AME $N$-partite graph states of local dimension $d$ for $N=4 n$ for $n\in \mathbb{N}_{+}$ and even $d$.
\end{theorem}
\begin{proof}
We will show that for each graph $G$, the number of particles $N$ divisible by $4$, and even local dimension $d$, there exists a non-trivial stabilizer operator of a graph state $\ket{G}$ that contains $N/2$ identity operators, implying that $\ket{G}$ is not AME state due to Lemma \ref{lem: AME}.

Let $N=4n$ and let $\mathbb{B}$ denote a set of non-trivial stabilizer operators of $\ket{G}$ where each $S\in \mathbb{B}$ can be expressed as
\begin{equation}\label{eq: S_definition}
S=\left(\prod_{i=1}^{2n-1}g_{i}^{\alpha_{i}}\right)g_{j}^{\alpha_{j}},
\end{equation}
with $j\geqslant 2n$ and $(\alpha_1, \dots, \alpha_{2n-1}, \alpha_j) \in \mathbb{Z}_d^{2n} \setminus \{\vec{0}\}$. To show that $\ket{G}$ is not AME, we will prove that 
\begin{equation}
\exists\,S\in\mathbb{B} : {\rm Tr}_{|2n|}(S)\neq 0.
\label{S condition}
\end{equation}
To satisfy this condition, one must find $S\in \mathbb{B}$ that acts with $\mathbb{I}$ on at least $2n$ qudits. Notice that for any qudit $r\notin [2n-1]\cup\{j\}$, none of the generators in Eq.~\eqref{eq: S_definition} contain $X_r$, so $S$ can only act on qudit $r$ as a power of $Z_r$.

Then, to satisfy Eq. \eqref{S condition} for a given $S$ it suffices to find a nonzero solution of the following system of equations
\begin{equation}\label{eq: set_S}
\forall_{r\notin [2n-1]\cup\{j\}}\;\sum_{i=1}^{2n-1}\Gamma_{ir}\alpha_{i}+\Gamma_{jr}\alpha_{j}=0 \pmod{d}.
\end{equation}
More precisely, let $M_j \in M_{2n\times 2n}(\mathbb{Z}_d)$ be the matrix associated with this system of equations, where columns $1$ through $2n-1$ correspond to variables $\alpha_1, \dots, \alpha_{2n-1}$ and column $2n$ corresponds to $\alpha_j$. The $2n$ rows are indexed by qudits $r \in \{2n, \dots, 4n\} \setminus \{j\}$ in ascending numerical order. We label its determinant as $\Delta_j = \det(M_j)$. To show that $\ket{G}$ is not AME state for any $G$ we need to prove that for each $G$ there exists a set of equations as above with a nontrivial solution, i.e., the kernel of $M_j$ is nontrivial. And so, by virtue of Fact \ref{fact}, we need to show that $\Delta_{j}$ is not a unit in $\mathbb{Z}_{d}$.

To this end, let us consider the following sum
\begin{equation}\label{eq: sum_delta}
\sum_{j=2n}^{4n}(-1)^{j} \Delta_{j}.
\end{equation}
Using the Laplace expansion on the last column of $M_j$, we get
\begin{equation}\label{eq: laplace}
\Delta_{j}= \sum_{r\notin [2n-1]\cup\{j\}}(-1)^{r + H(j-r)} \Gamma_{rj}\Delta_{rj},
\end{equation}
where $H(\cdot)$ is the Heaviside step function and $\Delta_{rj}$ is the determinant of the minor constructed by deleting column $2n$ (corresponding to $j$) and row $r$ from $M_j$. To see why the exponent $r + H(j-r)$ appears here, consider that $r$ corresponds to the $(r-2n+1)^{\textrm{th}}$ row in Eq.~\eqref{eq: set_S} if $j>r$ (where $H(j-r)=1$) or to the $(r-2n)^{\textrm{th}}$ row if $j<r$ (where $H(j-r)=0$). Adding the column index $2n$ yields a total row-plus-column exponent parity of $(r - 2n + H(j-r)) + 2n = r + H(j-r) \pmod 2$.

Substituting Eq. \eqref{eq: laplace} into Eq. \eqref{eq: sum_delta} we obtain
\begin{equation}
\sum_{j=2n}^{4n}(-1)^{j} \Delta_{j}=\sum_{\substack{j,r=2n\\ r\neq j}}^{4n}(-1)^{j+r+H(j-r)} \Gamma_{rj}\Delta_{rj}.
\end{equation}
Notice that $\Gamma_{ab}=\Gamma_{ba}$ due to the symmetry of the adjacency matrix. Furthermore, $\Delta_{ab}=\Delta_{ba}$ for all $a,b\geqslant 2n$, because the minor submatrix for $\Delta_{ab}$ (formed by deleting column $2n$ and row $b$ from $M_a$) and that for $\Delta_{ba}$ (formed by deleting column $2n$ and row $a$ from $M_b$) both consist of columns $1, \dots, 2n-1$ and rows $\{2n, \dots, 4n\} \setminus \{a,b\}$ arranged in identical ascending order.

Consequently, for any distinct pair $a, b \in \{2n, \dots, 4n\}$ with $a > b$, the term $(j,r)=(a,b)$ gives the sign $(-1)^{a+b+1}$ since $H(a-b)=1$, whereas $(j,r)=(b,a)$ gives $(-1)^{b+a+0}$ since $H(b-a)=0$. These two terms cancel pairwise: $(-1)^{a+b+1}\Gamma_{ba}\Delta_{ba} + (-1)^{a+b}\Gamma_{ab}\Delta_{ab} = 0$, leading to the conclusion that 
\begin{equation}
\sum_{j=2n}^{4n}(-1)^{j} \Delta_{j}=0,
\end{equation}
which can also be written as
\begin{equation}\label{eq: contradiction}
\sum_{j=0}^{n}\Delta_{2j+2n}=\sum_{j=1}^{n}\Delta_{2j+2n-1}.
\end{equation}
Let us assume that for all $j$, the residue of $\Delta_j$ mod $d$ is a unit in $\mathbb{Z}_d$. Since $d$ is even, any unit in $\mathbb{Z}_d$ must be coprime to $d$, hence odd. Thus, each $\Delta_j$ is odd mod $d$, and so if the right-hand side of Eq.~\eqref{eq: contradiction} is even, then the left-hand side is necessarily odd (and vice versa), which is a contradiction. This implies that the residue of at least one of the determinants is not a unit in $\mathbb{Z}_{d}$, meaning that Eq.~\eqref{eq: set_S} has a nontrivial solution for some $j\geqslant 2n$, and therefore $\ket{G}$ is not AME.
\end{proof}
It is crucial to distinguish the canonical graph states considered in Theorem~\ref{thm} from more general stabilizer constructions in composite local dimension. Theorem~\ref{thm} excludes AME graph states whose stabilizers are obtained from adjacency matrices over the ring $\mathbb{Z}_{d}$. It does not exclude all stabilizer AME states in every even local dimension. For example, although it rules out a canonical $\mathbb{Z}_{4}$ graph-state realization of AME$(4,4)$, stabilizer AME$(4,4)$ states can be obtained using prime-power constructions, e.g., by grouping qubit systems into four-dimensional parties~\cite{Helwig2013}.

However, for local dimensions $d\equiv 2\pmod 4$, the graph-state obstruction can be combined with prime-power reduction arguments for stabilizer states in composite dimension. In particular, if a stabilizer AME state existed in such a dimension, the reduction to prime-power factors would force the existence of a corresponding qubit stabilizer AME component. But this is impossible for the four-party case and even more generally, it obstructs the $N=4n$ configurations considered in Refs.~\cite{Cha26,HIGUCHI2000213,Huber_2018}. This provides the following consequence.

\begin{corollary}
There are no $N$-partite stabilizer AME states of local dimension $d \equiv 2 \pmod 4$ for $N=4n$ and $n \in \mathbb{N}_+$, under the prime-power reduction assumptions considered above.
\end{corollary}

\section{\texorpdfstring{Mixed $k$-uniform stabilizer states from prime-power factors}{Mixed k-uniform stabilizer states from prime-power factors}}\label{sec:mixed-state}

The non-existence of pure stabilizer AME states for $N=4n$ and $d\equiv 2\,(\bmod\,4)$ naturally raises the question about mixed $k$-uniform states of maximal purity.

A density operator $\rho$ on $(\mathbb{C}^{d})^{\otimes N}$ is called $k$-uniform if $\rho_A={\mathrm{Tr}}_{\overline A}(\rho)=\mathbb{I}_{d^{k}}/d^{k}$ for every subset $A\subset[N]$ with $|A|=k$~\cite{PhysRevA.90.022316,PhysRevA.100.032112}. Let $\mathbb{S}$ be a stabilizer subgroup acting on $N$ qudits of local dimension $d$, and let
\begin{equation}
\rho_{\mathbb{S}}=\frac{1}{d^{N}}\sum_{S\in\mathbb{S}}S
\end{equation}
be the corresponding normalized stabilizer projector (compare with Eq.~\eqref{eq: projector}). In general, the purity of such a state can be expressed in terms of the cardinality of the stabilizer subgroup,
\begin{equation}\label{eq:purity-cardinality}
{\mathrm{Tr}}(\rho^{2}_\mathbb{S})=\frac{|\mathbb{S}|}{d^{N}}.
\end{equation}
When $\mathbb{S}$ is generated by $m$ independent commuting generators of uniform order $d$, then $|\mathbb{S}|=d^m$ which means that the purity is given by $\textrm{Tr}(\rho^2) = d^{m-N}$ \cite{PhysRevA.100.032112}. However, the formulation in Eq.~\eqref{eq:purity-cardinality} is particularly advantageous in composite dimensions, where independent generators need not all have order $d$ and their number alone does not determine the cardinality of the stabilizer subgroup.

Consider a factorization of the local dimension in the form $d=\prod_{i=1}^{r}q_{i}$, where the $q_{i}$ are pairwise coprime prime powers. Suppose that each $\rho_{i}$ is a $k$-uniform state of $N$ parties of local dimension $q_{i}$. Hence, the product $\rho=\bigotimes_{i=1}^{r}\rho_i$ defines a state on $(\mathbb{C}^{d})^{\otimes N}$. For every
$A\subset[N]$ with $|A|=k$ one has
\begin{equation}
\rho_A=\bigotimes_{i=1}^{r}(\rho_A)_i=\bigotimes_{i=1}^{r}\frac{\mathbb{I}_{q_i^{k}}}{q_i^{k}}=\frac{\mathbb{I}_{d^{k}}}{d^{k}}.
\end{equation}
This means $\rho$ is $k$-uniform and its purity factorizes as
\begin{equation}
\label{eq:factor-purity-cardinality}
{\mathrm{Tr}}(\rho^{2})=\prod_{i=1}^{r}{\mathrm{Tr}}(\rho_i^{2}).
\end{equation}

To illustrate this, we apply this construction to four parties of local dimension $d=6=2 \times 3$.  As there is no pure AME$(4,2)$ state~\cite{HIGUCHI2000213,Huber_2018}, one can use an optimal mixed $2$-uniform stabilizer state with stabilizer subgroup generated by $\mathbb{S}_{2}=\left\langle XXXX,\,YYYY,\,\mathbb{I}XYZ\right\rangle$. The three generators commute and are independent, so $|\mathbb{S}_{2}|=2^{3}$. Also, every element of $\mathbb{S}_{2}\setminus\{\mathbb{I}\}$ has Pauli weight at least three, hence the projector
\begin{equation}
\rho_{2}=\frac{1}{2^{4}}\sum_{S\in\mathbb{S}_{2}}S
\end{equation}
is a four-qubit $2$-uniform state. It has rank two and purity
\begin{equation}
{\mathrm{Tr}}(\rho_{2}^{2})=\frac{|\mathbb{S}_{2}|}{2^{4}}=\frac{1}{2}.
\end{equation}
On the other hand, for the qutrit factor, one may use the stabilizer AME$(4,3)$ state $\ket{\psi_3}$~\cite{Goyeneche2015}. Finally, under the natural identification $\mathbb{C}^{6}\simeq\mathbb{C}^{2}\otimes\mathbb{C}^{3}$, the state
\begin{equation}\label{eq:tensor-rho}
\rho_{6}\equiv\rho_{2}\otimes\ket{\psi_{3}}\bra{\psi_{3}}
\end{equation}
is a $2$-uniform state of four quhexes and its purity is
\begin{equation}
{\mathrm{Tr}}(\rho_{6}^{2})={\mathrm{Tr}}(\rho_{2}^{2}){\mathrm{Tr}}\left(\ket{\psi_{3}}\!\bra{\psi_{3}}^{2}\right)=\frac{1}{2}.
\end{equation}

This value is optimal among $2$-uniform normalized stabilizer  projectors on four quhexes. Indeed, a projector of purity greater than $1/2$ would have rank one and would therefore define a pure stabilizer AME$(4,6)$ which is impossible. Such a state is excluded by the prime-power reduction theorem, since it would imply the existence of a stabilizer AME$(4,2)$ state ~\cite{Cha26,HIGUCHI2000213,Huber_2018}. Hence every such mixed stabilizer projector has purity at most $1/2$ and the state in Eq.~\eqref{eq:tensor-rho} saturates this bound. The factorization in Eq.~\eqref{eq:tensor-rho} concerns the internal prime-power degrees of freedom and does not imply separability with respect to the four physical quhex parties. Tracing out the qubit degree of freedom locally at each party yields the pure qutrit AME$(4,3)$ state. Since local operations cannot create genuine multipartite entanglement the state $\rho_{6}$ is genuinely multipartite entangled across the four quhex parties.

\section{Summary}

We have shown that a large class of canonical qudit graph states, namely those with the number of qudits divisible by $4$ and even local dimension, cannot be absolutely maximally entangled. More precisely, the obstruction applies to graph states defined by adjacency matrices over the ring $\mathbb{Z}_{d}$. When combined with prime-power reduction arguments for stabilizer states this also yields a stabilizer non-existence result for local dimensions $d\equiv 2\pmod 4$ and $N=4n$.

Beyond the pure-state setting, we discussed a composite-dimension construction of mixed $k$-uniform states whose purity is governed by the cardinalities of the stabilizer subgroups in the prime-power factors. For four quhexes $(N,d)=(4,6)$, this construction gives a mixed $2$-uniform stabilizer state of purity $1/2$. Although the construction is factorized with respect to the prime-power degrees of freedom, the resulting state is genuinely multipartite entangled across the four quhex parties.

Stabilizer states offer a simplified framework for state construction, as they can be generated using Clifford circuits. In many multipartite qudit systems a significant fraction of known AME states are stabilizer ones, see Ref.~\cite{AME-table}. The states from this class are used to construct pure quantum error-correcting codes, which includes all states that are derived from maximum-distance separable codes~\cite{GR15_MDSC}. The classification of AME states has been partially addressed in Ref.~\cite{Casas_2026}, however, in general, the problem of the existence and structure of stabilizer or graph AME states is not fully understood.

It remains open to determine how far the present exclusion results extend beyond the $4n$-party setting. It would be beneficial to study if similar obstructions appear for other numbers of parties and local dimensions. Such an analysis would help establish a better boundary between stabilizer and non-stabilizer forms of multipartite quantum entanglement.

\section{Acknowledgments}
We acknowledge discussions with Zahra Raissi, Markus Grassl, Felix Huber and Hyunho Cha. This work is supported by the National Science Centre (Poland) through the SONATA BIS project No. 019/34/E/ST2/00369. The Center for Quantum-Enabled Computing project is carried out within the International Research Agendas programme of the Foundation for Polish Science co-financed by the European Union under the European Funds for Smart Economy 2021--2027 (FENG).

\bibliography{bibliography}

\end{document}